%% file: paper.tex
\documentclass[submission,copyright,creativecommons
]{eptcs}
\providecommand{\event}{10th International Workshop on Fixed Points in Computer Science (FICS 2015)} 

\usepackage{amsmath}
\usepackage{amssymb}
\usepackage{amsthm}
\usepackage{amsfonts} 
\usepackage{latexsym}
\usepackage{hyperref}

\usepackage{fics2015}

\title{Formalizing Termination Proofs under Polynomial Quasi-interpretations
}
\author{Naohi Eguchi%
\thanks{The author is supported by Grants-in-Aid for JSPS Fellows (Grant No.
$25 \cdot 726$).}
\institute{Department of Mathematics and Informatics \\
Chiba University,
Japan%
}
\email{neguchi@g.math.s.chiba-u.ac.jp}
}
\def\titlerunning{Formalizing Termination Proofs under Polynomial Quasi-interpretations}
\def\authorrunning{Naohi Eguchi}

\begin{document}

\maketitle

\theoremstyle{plain}
\newtheorem{theorem}{Theorem}
\newtheorem{corollary}{Corollary}
\newtheorem{lemma}{Lemma}
\newtheorem{claim}{Claim}
\newtheorem{proposition}{Proposition}

\theoremstyle{definition}
\newtheorem{definition}{Definition}

\theoremstyle{remark}
\newtheorem{remark}{Remark}
\newtheorem{example}{Example}

\begin{abstract}
Usual termination proofs for a functional program require to check all
 the possible reduction paths.
Due to an exponential gap between the height and size of such the
 reduction tree, no naive formalization of termination proofs yields a
 connection to the polynomial complexity of the given program.
We solve this problem employing the notion of minimal function
 graph, a set of pairs of a term and its normal form, which is defined
 as the least fixed point of a monotone operator.
We show that termination proofs for programs reducing under
 lexicographic path orders (LPOs for short) and polynomially
 quasi-interpretable can be optimally performed in a weak fragment of
 Peano arithmetic.
This yields an alternative proof of the fact that every function
 computed by an LPO-terminating, polynomially quasi-interpretable
 program is computable in polynomial space.
The formalization is indeed optimal since every polynomial-space computable
 function can be computed by such a program.
The crucial observation is that inductive definitions of minimal
 function graphs under LPO-terminating programs can be approximated with
 transfinite induction along LPOs.
\end{abstract}

\input{intro}
\input{prelim}
\input{lpo_pqi}
\input{bndd_arith}
\input{mfg}
\input{termination}
\input{application}

\section{Conclusion}

This work is concerned with optimal termination proofs for functional programs
in the hope of establishing logical foundations of
computational resource analysis.
Optimal termination proofs were limited for programs that compute
functions lying in complexity classes closed under exponentiation.
In this paper, employing the notion of minimal function graph, we showed
that termination proofs under $\LPOPoly$-programs can be optimally formalized in
the second order system $\mU^1_2$ of bounded arithmetic
that is complete for
polynomial-space computable functions, lifting the limitation.
The crucial idea is that inductive definitions of minimal function
graphs under $\LPOPoly$-programs can be approximated with transfinite
induction along LPOs.
As a small consequence, compared to the original result, Theorem
\ref{t:pqi}, when we say ``a program $\RS$ computes a function'',
the quasi-reducibility of $\RS$ is explicitly needed to enable the formalization.

Finally, let us call a program $\RS$ an
{\em $\MPOPoly$} one if $\RS$ reduces under an MPO (with product
status only) and $\RS$ admits a kind $0$ PQI.
In \cite[Theorem 42]{BMM11}, Theorem \ref{t:pqi} is refined so that a
function can be computed by an
$\MPOPoly$-program if and only if it is computable in polynomial time.
The program $\RS_{\m{lcs}}$ described in Example \ref{ex:LCS} is an example of $\MPOPoly$-programs, and
hence the length of the longest common subsequences
is computable even in polynomial time.
By Theorem \ref{t:Buss86}.\ref{t:Buss86:S}, it is quite natural to expect that minimal
function graphs under $\MPOPoly$-programs can be constructed in the
first order system $\mS^1_2$.
However, we then somehow have to adopt the formula 
$\varphi_{\ell} (t, s) \equiv
  \QR{\ell} \wedge
  \LPO \wedge
  \PQI{\ell} 
  \rightarrow 
  \exists G \ \psi_{\ell} (t, s, G)
$
(in the proof of Corollary \ref{c:pqi}) to a $\bSig{1}$-formula, which
is clearly more involved than the present case.  

\nocite{*}
\bibliographystyle{eptcs}
\bibliography{fics2015}

%
\end{document}

%% file: intro.tex
\section{Introduction}

\subsection{Motivation}
\label{ss:motivation}

The termination of a program states that any reduction under the program leads to a normal form.
Recent developments in termination analysis of first order functional
programs, or of {\em term rewrite systems} more specifically, have drawn
interest in computational resource analysis, i.e., not just the
termination but also the estimation of
time/space-resources required to execute a given program,
which includes the polynomial run-space complexity analysis.
Usual termination proofs for a program require to check all the
possible reduction paths under the program.
Due to an exponential gap between the {\em height} and {\em size} of such the
reduction tree, no naive termination proof yields a connection to the
polynomial complexity of the given program.
For the sake of optimal termination proofs, it seems necessary to discuss ``all the possible reduction paths'' by means of an alternative notion smaller in size than reduction trees.

\subsection{Backgrounds}
\label{ss:background}

Stemming from \cite{Thompson72}, there are various functional characterizations of polynomial-space computable functions
\cite{LM95,Oitavem01,Oitavem02,Eguchi10},
Those characterizations state that every poly-space computable function can be defined by a finite set of equations, i.e., by a functional program.
Orienting those equations suitably, such programs reduce under a termination order, the
{\em lexicographic path orders} (LPOs for short).
The well-founded-ness of LPOs yields the termination of the reducing programs.

In the seminal work \cite{buch95}, it was discussed, depending on the choice of a
termination order, what mathematical axiom is necessary to formalize termination proofs by the termination order within Peano arithmetic $\mathrm{PA}$ that axiomatizes ordered semi-rings
with mathematical induction.
In case of
{\em multiset path orders} (MPOs for short),
termination proofs can be formalized in the fragment of $\mathrm{PA}$
with induction restricted to computably enumerable sets.
This yields an alternative proof of the fact that every function computed by an MPO-terminating program is primitive recursive, cf. \cite{Hof90}.
The formalization is optimal since every primitive recursive function can be computed by an MPO-terminating program.
In case of LPOs,
termination proofs can be formalized in the fragment with
induction restricted to expressions of the form ``$f$ is total'' for
some computable function $f$.
The formalization is optimal in the same sense as in case of MPOs, cf. \cite{Weier95}.

In more recent works \cite{BMM01,BMM11}, MPOs and LPOs are combined with
{\em polynomial quasi-interpretations} (PQIs for short).
Unlike (strict) polynomial interpretations \cite{BCMT01}, the existence of a quasi-interpretation does not tell us anything about termination.
However, combined with these termination orders, the PQI can be a powerful method in computational resource analysis.
Indeed, those functional programs characterizing poly-space computable functions that was motioned above admit PQIs.
This means that every poly-space computable function can be computed by an LPO-terminating program that admits a PQI.
Moreover, conversely, every function computed by such a program is
computable in polynomial space \cite[Theorem 1]{BMM01}.

\subsection{Outline}
\label{s:outline}

In Section \ref{s:prelim} we fix the syntax of first order functional
programs and the semantics in accordance with the syntax.
In Section \ref{s:lpo_pqi} we present the definitions of LPOs and PQIs
together with some examples, stating an application to poly-space
computable functions (Theorem \ref{t:pqi}, \cite[Theorem 1]{BMM01}).
In Section \ref{s:bndd_arith} we present the framework of formalization.
For an underlying formal system, a second order system $\mU^1_2$ of
{\em bounded arithmetic} \cite{buss86}, which can be regarded as a weak fragment of
$\mathrm{PA}$, seems suitable since it is known that the system $\mU^1_2$ is
complete for poly-space computable functions (Theorem
\ref{t:Buss86}.\ref{t:Buss86}).

In \cite{buch95}, the termination of a program reducing under an LPO
$<_{\lpo}$ is deduced by showing that, given a term $t$, a tree containing all the possible reduction chains starting with $t$ is well founded under $<_{\lpo}$.
The same construction of such reduction trees does not work in $\mU^1_2$
essentially because the exponentiation $m \mapsto 2^m$ is not available.
We lift the problem employing the notion of
{\em minimal function graph} \cite{JonesM86,Jones97,Marion03}, a set of pairs of
a term and its normal form.
Given a term $t$, instead of constructing a reduction tree rooted at
$t$, we construct a (subset of a) minimal function graph that stores
the pair of $t$ and a normal form of $t$.
Typically, a minimal function graph is inductively defined, or in other
words defined as the least fixed point
of a monotone operator.
Let us recall that the set of natural numbers is the least fixed point of the
operator
$m \in \Gamma (X) \Longleftrightarrow
 m=0 \vee \exists n \in X \text{ s.t. } m = n+1$.
As seen from this example, many instances of inductive definitions are
induced by operators of the form
$t \in \Gamma (X) \Longleftrightarrow
 \exists s_1, \dots, s_k \in X \cdots$.
Crucially, a minimal function graph under a program reducing under an
LPO $<_{\lpo}$ can be defined as the least fixed point of
such an operator but also
$t \in \Gamma (X) \Longleftrightarrow
 \exists s_1, \dots, s_k \in X \wedge s_1, \dots, s_k <_{\lpo} t \cdots$
holds.
Thanks to the additional condition
$s_1, \dots, s_k <_{\lpo} t$, the minimal function graphs under the
program can be defined by $<_{\lpo}$-transfinite induction as well as
inductive definitions.
In Section \ref{s:mfg} this idea is discussed in more details.

In the main section, Section \ref{s:termination}, the full details about
the formalization are given.
Most of the effort is devoted to deduce in $\mU^1_2$ an appropriate form of
transfinite induction along LPOs (Lemma \ref{l:TI:B}).
Based on the idea above, we then construct a minimal function graph $G$ for
a given program $\RS$ reducing under an LPO $<_{\lpo}$ by
$<_{\lpo}$-transfinite induction (Theorem \ref{t:MFG}).
Since $G$ stores all the pairs of a term and its $\RS$-normal form, this means
the termination of the program $\RS$.

In Section \ref{s:application} it is shown that the formalization presented
in Section \ref{s:termination} yields that every function
computed by an LPO-terminating program that admits a PQI is poly-space
computable (Corollary \ref{c:pqi}).
This shows that the formalization is optimal since such programs can only compute
poly-space computable functions as mentioned in Section \ref{ss:background}.

%% file: prelim.tex
\section{Syntax and semantics of first order functional programs}
\label{s:prelim}

Throughout the paper, a {\em program} denotes a
{\em term rewrite system}.
We sometimes use unusual notations or formulations for the sake of simplification.
More precise, widely accepted formulations can be found, e.g., in \cite{Terese}.
\begin{definition}[Constuctor-, basic-, terms, rewrite rules, sizes of terms]
Let $\CS$ and $\DS$ be disjoint finite signatures, respectively of
{\em constructors} and {\em defined} symbols, and
$\VS$ a countably infinite set of {\em variables}.
We assume that $\CS$ contains at least one constant.
The sets
$\mathbf{T} (\CS \cup \DS,\VS)$ of {\em terms},
$\mathbf{T} (\CS, \VS)$ of {\em constructor} terms,
$\mathbf{B} (\CS \cup \DS, \VS)$ of {\em basic} terms and
$\RS (\CS \cup \DS, \VS)$ of {\em rewrite rules} are distinguished as follows.
\[
\begin{array}{lrcll}
 (\text{Terms}) & t & ::= &
  x \mid c (t_1, \dots, t_l) \mid f(t_1, \dots, t_l) &
  \in \mathbf{T(C \cup D,V)}; \\
 (\text{Constructor terms}) & s & ::= &
  x \mid c(s_1, \dots, s_k) & \in \mathbf{T(C,V)}; \\
 (\text{Basic terms}) & u & ::= &
  f(s_1, \dots, s_k) & \in \mathbf{B(C \cup D,V)}; \\
(\text{Rewrite rules}) & \rho & ::= &
  u \rightarrow t & \in \mathbf{R(C \cup D, V)},
\end{array}
\]
where $x \in \VS$, $c \in \CS$, $f \in \DS$,
$t, t_1, \dots, t_l \in \mathbf{T(C \cup D, V)}$,
$s_1, \dots, s_k \in \mathbf{T(C,V)}$
 and $u \in \mathbf{B(C \cup D,V)}$.
For such a class $\mathbf{S(F,V)}$ of terms,
$\mathbf{S(F)}$ denotes the subset of closed terms.
The {\em size} $\size{t}$ of a term $t$ is defined as
$\size{x} = 1$ for a variable $x$ and
$\size{f(t_1, \dots, t_k)} = 1 + \sum_{j=1}^k \size{t_j}$.
\end{definition}

\begin{definition}[Substitutions, quasi-reducible programs, rewrite relations]
A {\em program} $\RS$ is a finite subset of $\mathbf{R(C \cup D, V)}$
consisting of rewrite rules of the form $l \rightarrow r$
such that the variables occurring in $r$ occur in $l$ as well.
A mapping
$\theta: \VS \rightarrow \mathbf{S(F,V)}$
from variables to a set $\mathbf{S(F,V)}$ of terms
is called a {\em substitution}.
For a term $t \in \mathbf{S(F,V)}$,
$t \theta$ denotes the result of replacing every variable $x$ with
$\theta (x)$.
A program $\RS$ is {\em quasi-reducible} if, for any closed basic term
$t \in \mathbf{B(C \cup F)}$, there exist a rule
$l \rightarrow r \in \RS$ and a substitution
$\theta: \VS \rightarrow \TC{}$ such that
$t = l \theta$.
We restrict reductions to those under
{\em call-by-value} evaluation, or {\em innermost} reductions more
precisely.
For three terms $t, u, v$, we write $t[u/v]$ to denote the result
of replacing an occurrence of $v$ with $u$.
It will not be indicated which occurrence of $v$ is replaced if no confusion
likely arises.
We write
$t \rew s$ if
$s = t[r \theta / l \theta]$ holds for some rule
$l \rightarrow r \in \RS$ and constructor substitution
$\theta: \VS \rightarrow \TC{}$.
We write $\rewast$ to denote the reflexive and transitive closure of
$\rew$ and
$t \rewnf s$ if $t \rewast s$ and $s$ is a normal form.
By definition, for any quasi-reducible program $\RS$,
if $t \rewnf s$ and $t$ is closed, then
$s \in \TC{}$ holds.
\end{definition}

A program $\RS$ {\em computes} a function if any closed basic term has a
unique normal form in $\TC{}$.
In this case, for every $k$-ary function symbol $f \in \DS$,
a function
$\fint{f}: \TC{}^k \rightarrow \TC{}$
is defined by
$\fint{f} (s_1, \dots, s_k) = s \Longleftrightarrow
 f(s_1, \dots, s_k) \rewnf s$.

%% file: lpo_pqi.tex
\section{Lexicographic path orders and quasi-interpretations}
\label{s:lpo_pqi}

Lexicographic path orders are {\em recursive path orders} with lexicographic
status only, whose variant was introduced in \cite{Kamin_Levy}.
Recursive path orders with multiset status only were introduced in
\cite{Dershowitz82} and a modern formulation with both multiset and lexicographic status can be found in
\cite[page 211]{Terese}.
Let $<_{\FS}$ be a (strict) {\em  precedence}, a well-founded partial order
on a signature $\FS = \CS \cup \DS$.
We always assume that every constructor is $<_{\FS}$-minimal.
The {\em lexicographic path order} (LPO for short)
$<_{\lpo}$ induced by $<_{\FS}$ is defined recursively by the
following three rules.
\label{page:lpo}
\begin{enumerate}
\item $\displaystyle
       \frac{s \leqslant_{\lpo} t_i}{s <_{\lpo} g(t_1, \dots, t_l)}$
      $(i \in \{ 1, \dots, l \})$
\label{d:lpo:1}
\item $\displaystyle
       \frac{
             s_1 <_{\lpo} g(t_1, \dots, t_l) \quad \cdots \quad
             s_k <_{\lpo} g(t_1, \dots, t_l)
            }%
            {f(s_1, \dots, s_k) <_{\lpo} g(t_1, \dots, t_l)}
       \ (f <_{\FS} g \in \DS)$
\label{d:lpo:2}
\item $\displaystyle
       \frac{
             s_1 = t_1 \ \cdots \
             s_{i-1} = t_{i-1} \quad
             s_i <_{\lpo} t_i \quad
             s_{i+1} <_{\lpo} t \ \cdots \
             s_{k} <_{\lpo} t
            }%
            {f(s_1, \dots, s_k) <_{\lpo} f(t_1, \dots, t_k) =t}
       \ (f \in \DS)$
\label{d:lpo:3}
\end{enumerate}

We say that a program $\RS$
{\em reduces under} $<_{\lpo}$ if
$r <_{\lpo} l$ holds for each rule $l \rightarrow r \in \RS$
 and that $\RS$ is
{\em LPO-terminating} if there exists an LPO under which $\RS$ reduces.
We write
$s <_{\lpo}^{\!\text{\tiny $\langle i \rangle$}} t$
if $s <_{\lpo} t$ results as an instance of the above
$i^{\text{th}}$ case
($i = \ref{d:lpo:1}, \ref{d:lpo:2}, \ref{d:lpo:3}$).
Corollary \ref{c:lpo} is a consequence of the definition of LPOs, following from 
$<_{\FS}$-minimality of constructors.
\begin{corollary}
\label{c:lpo}
If $s <_{\lpo} t$ and $t \in \TC{}$, then
$s \dlpo{\lpo}{1} t$ and
$s \in \TC{}$.
\end{corollary}

A {\em quasi-interpretation} $\qint{\cdot}$ for a signature
$\FS$ is a mapping from $\FS$ to
functions over naturals fulfilling
(i) $\qint{f}: \mathbb N^k \rightarrow \mathbb N$ for each
     $k$-ary function symbol $f \in \FS$,
(ii) $\qint{f} (\dots, m, \dots) \leq
      \qint{f} (\dots, n, \dots)$
     whenever $m < n$,
(iii) $m_j \leq \qint{f} (m_1, \dots, m_k)$ for any
     $j \in \{ 1, \dots, k \}$, and
(iv) $0 < \qint{f}$ if $f$ is a constant.
A quasi-interpretation $\qint{\cdot}$ for a signature $\FS$ is extended to
closed terms $\mathbf{T(F)}$ by
$\qint{f(t_1, \dots, t_k)} =
 \qint{f} (\qint{t_1}, \dots, \qint{t_k})$.
Such an interpretation $\qint{\cdot}$ is called a quasi-interpretation for a
program $\RS$ if
$\qint{r \theta} \leq \qint{l \theta}$ 
holds for each rule $l \rightarrow r \in \RS$ and
for any constructor substitution
$\theta: \VS \rightarrow \TC{}$.
A program $\RS$
{\em admits a polynomial quasi-interpretation}
(PQI for short) if there
exists a quasi-interpretation $\qint{\cdot}$ for $\RS$ such that 
$\qint{f}$ is polynomially bounded for each $f \in \FS$.
A PQI $\qint{\cdot}$ is called {\em kind $0$}
(or {\em additive} \cite{BMM11})  if,
for each constructor $c \in \CS$,
$\qint{c} (m_1, \dots, m_k) =
 d + \sum_{j=1}^k m_j$
holds for some constant $d > 0$.
An LPO-terminating program $\RS$ is called an
{\em $\LPOPoly$}-program if
$\RS$ admits a kind $0$ PQI.

\begin{theorem}[\cite{BMM01}]
\label{t:pqi}
Every function computed by an $\LPOPoly$-program is computable in polynomial space.
\end{theorem}
Conversely, every polynomial-space computable function can be computed by an
$\LPOPoly$-program \cite[Theorem 1]{BMM01}.
In \cite{BMM11} various examples of programs admitting (kind $0$) PQIs are illustrated, including $\LPOPoly$-programs
$\RS_{\m{lcs}}$ and $\RS_{\m{QBF}}$ below.
\begin{example}
\label{ex:LCS}
The length of the {\em longest common subsequences} of two strings can
 be computed by a program $\RS_{\m{lcs}}$ \cite[Example 6]{BMM11},
which consists of the following rewrite rules defined over a signature
 $\FS = \CS \cup \DS$ where
 $\CS = \{ \m{0}, \m{s}, \epsilon, \m{a}, \m{b} \}$
 and
 $\DS = \{ \m{max}, \m{lcs} \}$.
\[
 \begin{array}[t]{rclrcll}
  \m{max} (x, \m{0}) & \rightarrow & x &
   \m{max} (\m{s} (x), \m{s} (y)) & \rightarrow &
   \m{s} (\m{max} (x, y)) & \\
  \m{max} (\m{0}, y) & \rightarrow & y & & & & \\
  \m{lcs} (x, \epsilon) & \rightarrow & \m{0} &
   \m{lcs} (\m{i} (x), \m{i} (y)) & \rightarrow &
   \m{s} (\m{lcs} (x, y)) & (\m{i} \in \{ \m{a}, \m{b} \}) \\
  \m{lcs} (\epsilon, y) & \rightarrow & \m{0} &
   \m{lcs} (\m{i} (x), \m{j} (y)) & \rightarrow &
   \m{max} (\m{lcs} (x, \m{j} (y)), \m{lcs} (\m{i} (x), y)) &
   (\m{i} \neq \m{j} \in \{ \m{a}, \m{b} \} )
 \end{array}
\]
Natural numbers are built of $\m{0}$ and $\m{s}$ and
strings of $\m{a}$ and $\m{b}$ as
$\m{a} (u) = \m{a} u$
for a string $u \in \{ \m{a}, \m{b} \}^\ast$.
The symbol $\epsilon$ denotes
 the empty string.
 Define a precedence $<_{\FS}$ on $\FS$ by
 $\m{max} <_{\FS} \m{lcs}$.
 Assuming that every constructor is $<_{\FS}$-minimal, the program
 $\RS_{\m{lcs}}$ reduces under the LPO $<_{\lpo}$ induced by $<_{\FS}$.
 For instance, the orientation
 $\m{max} (\m{lcs} (x, \m{b} (y)), \m{lcs} (\m{a} (x), y))
 <_{\lpo} \m{lcs} (\m{a} (x), \m{b} (y))$
 can be deduced as follows.
 The orientation
 $y \dlpo{\lpo}{1} \m{b} (y)$
 yields
 $\m{lcs} (\m{a} (x), y) \dlpo{\lpo}{3} \m{lcs} (\m{a} (x), \m{b} (y))$
 while
 $x \dlpo{\lpo}{1} \m{a} (x)$ and
 $\m{b} (y) \dlpo{\lpo}{1} \m{lcs} (\m{a} (x), \m{b} (y))$
 yield
 $\m{lcs} (x, \m{b} (y)) \dlpo{\lpo}{3} \m{lcs} (\m{a} (x), \m{b} (y))$.
 These together with $\m{max} <_{\FS} \m{lcs}$ yield
 $\m{max} (\m{lcs} (x, \m{b} (y)), \m{lcs} (\m{a} (x), y))
  \dlpo{\lpo}{2} \m{lcs} (\m{a} (x), \m{b} (y))$.
 It can be seen that the program $\RS_{\m{lcs}}$ admits the kind $0$ PQI
 $\qint{\cdot}$ defined by
 \begin{eqnarray*}
 \qint{\m{0}} = \qint{\epsilon}& =& 1, \\
 \qint{\m{s}} (x) = \qint{\m{a}} (x) = \qint{\m{b}} (x)& =& 1+x, \\
 \qint{\m{max}} (x, y) = \qint{\m{lcs}} (x, y)& =& \max ( x, y).
 \end{eqnarray*}
 This is exemplified as
 $\qint{\m{max} (\m{lcs} (x, \m{b} (y)), \m{lcs} (\m{a} (x), y))}
  =
  \max \big( \max (x, 1+y), \max (1+x, y)
 \big)
  \leq
 \max (1+x, 1+y) =
  \qint{\m{lcs} (\m{a} (x), \m{b} (y))}$.
 Thus Theorem \ref{t:pqi} implies that the function
 $\fint{\m{lcs}}$ can be computed in polynomial space.
\end{example}
\begin{example}
\label{ex:QBF}  
The {\em Quantified Boolean Formula} (QBF) problem can be solved by
a program $\RS_{\m{QBF}}$
\cite[Example 36]{BMM11},
which consists of the following rewrite rules defined over a signature
 $\FS = \CS \cup \DS$ where
 $\CS = \{ \m{0}, \m{s}, \m{nil}, \m{cons},
           \m{\top}, \m{\bot}, \m{var},
           \neg, \m{\vee}, \m{\exists}
 \}$
 and
$\DS = \{ {=}, \m{not}, \m{or}, \m{in}, \m{verify}, \m{qbf} \}$.
\[
\begin{array}[t]{rclrcl}
 \m{not} (\top) & \rightarrow & \bot &
  \m{not} (\bot) & \rightarrow & \top \\
 \m{or} (\top, x) & \rightarrow & \top &
  \m{or} (\bot, x) & \rightarrow & x \\
 \m{0} = \m{0} & \rightarrow & \top &
  \m{s} (x) = \m{0} & \rightarrow & \bot \\
 \m{0} = \m{s} (x) & \rightarrow & \bot &
  \m{s} (x) = \m{s} (y) & \rightarrow & x=y \\
 \m{in} (x, \m{nil}) & \rightarrow & \bot &
  \m{in} (x, \m{cons} (y, ys)) & \rightarrow &
   \m{or} (x=y, \m{in} (x, ys))
\end{array} 
\]
\[
 \begin{array}[t]{rcl}
 \m{verify} (\m{var} (x), xs) & \rightarrow & \m{in} (x, xs) \\
 \m{verify} (\neg x, xs) & \rightarrow &
  \m{not} (\m{verify} (x, xs)) \\
 \m{verify} (x \vee y, xs) & \rightarrow &
  \m{or} \left( \m{verify} (x, xs), \m{verify} (y, xs)
         \right) \\
 \m{verify} \left( (\exists x) y, xs \right) & \rightarrow &
  \m{or} \left( \m{verify} (y, \m{cons} (x, xs)),
	        \m{verify} (y, xs)
		   \right) \\
 \m{qbf} (x) & \rightarrow & \m{verify} (x, \m{nil})
\end{array}
\]
The symbol $\top$ denotes the true Boolean value while $\bot$ the false
 one.
Boolean variables are encoded with
 $\{ \m{0}, \m{s} \}$-terms, i.e., with naturals.
Formulas are built from variables operating
 $\m{var}$, $\neg$, $\vee$ or $\exists$.
Without loss of generality, we can assume that every QBF is built up in
 this way.
 As usual, terms of the forms
 ${=} (s, t)$, $\neg (t)$, ${\vee} (s, t)$ and ${\exists} (s, t)$
 are respectively denoted as
 $s =t$, $\neg t$, $s \vee t$ and $(\exists s) t$. 
By definition, for a Boolean formula
 $\varphi$ with Boolean variables $x_1, \dots, x_k$,
 $\fint{\m{verify}} (\varphi, [\cdots]) = \top$
 holds if and only if
 $\varphi$ is true with the truth assignment that
 $x_j = \top$ if $x_j$ appears in the list $[\cdots]$ and
 $x_j = \bot$ otherwise.
 
 Define a precedence $<_{\FS}$ over
 $\FS$ by
 $\m{not}, \m{or}, {=} <_{\FS} \m{in}
  <_{\FS} \m{verify} <_{\FS} \m{qbf}$.
 Assuming $<_{\FS}$-minimality of constructor, the program
 $\RS_{\m{QBF}}$ reduces under the LPO
 $<_{\lpo}$ induced by $<_{\FS}$.
 For instance, the orientation
 $\m{or} (\m{verify} (y, \m{cons} (x, xs)),
          \m{verify} (y, xs)
         )
  <_{\lpo}
  \m{verify} (\exists (x, y), xs)
 $
 can be deduced as follows.
 As well as
 $xs \dlpo{\lpo}{1} \m{verify} (\exists (x, y), xs)$,
 the orientation
 $x \dlpo{\lpo}{1} \exists (x, y)$
 yields
 $x \dlpo{\lpo}{1} \m{verify} (\exists (x, y), xs)$.
 These together with the assumption
 $\m{cons} <_{\FS} \m{verify}$ yield
 $\m{cons} (x, xs) \dlpo{\lpo}{2}
 \m{verify} (\exists (x, y), xs)$.
 This together with
 $y \dlpo{\lpo}{1} \exists (x, y)$ yields
 $\m{verify} (y, \m{cons} (x, xs)) \dlpo{\lpo}{3}
 \m{verify} (\exists (x, y), xs)$
 as well as
 $\m{verify} (y, xs) \dlpo{\lpo}{3}
 \m{verify} (\exists (x, y), xs)$.
 These orientations together with the assumption
 $\m{or} <_{\FS} \m{verify}$
 now allow us to deduce the desired orientation
 $\m{or} (\m{verify} (y, \m{cons} (x, xs)),
          \m{verify} (y, xs)
         )
  \dlpo{\lpo}{2}
  \m{verify} (\exists (x, y), xs)
 $.
 
 Furthermore, let us define a PQI
 $\qint{\cdot}$ for the signature $\FS$ by
 \[
  \begin{array}[t]{rcll}
   \qint{c} &= & 1 & \text{if } c \text{ is a constant,} \\
   \qint{x_1, \dots, x_k}&= &
    1+ \sum_{j=1}^k x_j & \text{if } c \in \CS
    \text{ with arity} > 0, \\
   \qint{f} (x_1, \dots, x_k)&= & \max_{j=1}^k x_j &
    \text{if } f \in \DS \setminus
    \{ \m{verify}, \m{qbf} \}, \\
   \qint{\m{verify}} (x, y)&= &x+y, & \\
   \qint{\m{qbf}} (x)&= &x+1. & \\
  \end{array}
 \]
 Clearly the PQI $\qint{\cdot}$ is kind $0$.
 Then the program $\RS_{\m{QBF}}$ admits the PQI.
 This is exemplified by the rule above as
 $\qint{\m{or} (\m{verify} (y, \m{cons} (x, xs)),
                \m{verify} (y, xs)
               )
       } 
 = \max \big( y+ (1+x+xs), y+xs \big)
 = (1+x+y)+xs =
 \qint{\m{verify} (\exists (x, y), xs)}$.
 Thus Theorem \ref{t:pqi} implies that the function
 $\fint{\m{qbf}}$ can be computed in polynomial space.
 This is consistent with the well known fact that the QBF problem is PSPACE-complete.
 \end{example}

%% file: bndd_arith.tex
\section{A system $\mU^1_2$ of second order bounded arithmetic}
\label{s:bndd_arith}
	
In this section, we present the basics of second order bounded arithmetic following
\cite{BeckB14}.
The original formulation is traced back to \cite{buss86}.
The non-logical language $\LBA$ of first order bounded arithmetic consists of the constant 
$0$, the successor $\mrS$, the addition ${+}$, the multiplication ${\cdot}$,
$|x| = \lceil \log_2 (x+1) \rceil$, the division by two
$\divtwo{x}$, the smash $\# (x, y) = 2^{|x| \cdot |y|}$ and ${\leq}$.
It is easy to see that $|m|$ is equal to the number of bits in the binary
representation of a natural $m$.
In addition to these usual symbols, we assume that the language
$\LBA$ contains
$\max (x, y)$.
The assumption makes no change if an underlying system is
sufficiently strong.
\begin{definition}[Sharply-, bounded quantifiers, bounded formulas,
$\mS^1_2$]
Quantifiers of the form 
$\exists x (x \leq t \wedge \cdots)$ or
$\forall x (x \leq t \rightarrow \cdots)$
for some term $t$ are called
{\em bounded} and quantifiers of the form
$(Q x \leq |t|) \cdots$ are called
{\em sharply} bounded.
{\em Bounded formulas} contain no unbounded first order quantifiers.
The classes $\bSig{i}$ ($i \in \mathbb N$) of bounded formulas are defined by
counting the number of alternations of
bounded quantifiers starting with an existential one, but ignoring
sharply bounded ones.
For each $i \in \mathbb N$, the first order system $\mS^i_2$ of bounded
arithmetic is axiomatized with a set $\m{BASIC}$ of open axioms
defining the $\LBA$-symbols together with the schema
$(\bSig{i} \pind)$ of bit-wise induction 
for $\bSig{i}$-formulas.
\begin{equation}
\tag{$\Phi \pind$}
\label{e:PIND}
 \varphi (0) \wedge \forall x 
 \big( \varphi (\divtwo{x}) \rightarrow \varphi (x)
 \big) \rightarrow 
 \forall x \varphi (x)
\quad (\varphi \in \Phi)
\end{equation}
\end{definition}

The precise definition of the basic axioms $\m{BASIC}$ can be found,
e.g., in \cite[page 101]{buss98}.

\begin{definition}[Second order bounded formulas,
$\mU^1_2$]
In addition to the first order language, the language of second order
bounded arithmetic contains second order variables
$X, Y, Z, \dots$ ranging over sets and the membership relation
${\in}$.
In contrast to the classes $\bSig{i}$, the classes
$\BSig{i}$ of second order bounded formulas are defined by counting alternations of
second order quantifiers starting with an existential one, but ignoring
first order ones.
By definition, $\BSig{0}$ is the class of bounded formulas with no
 second order quantifiers.
The second order system $\mU^1_2$ is axiomatized with 
$\m{BASIC}$, $(\BSig{1} \pind)$ and the axiom
$(\BSig{0} \comp)$ of comprehension for $\BSig{0}$-formulas.
\begin{equation}
\tag{$\Phi \comp$}
\label{e:comp}
 \forall \vec x ~ \forall \vec X ~ \exists Y (\forall y \leq t)
 \big( y \in Y \leftrightarrow \varphi (y, \vec x, \vec X)
 \big)
\quad (\varphi \in \Phi)
\end{equation}
\end{definition}
Unlike first order ones, second order quantifiers have no
explicit bounding.
However, due to the presence of a bounding term $t$ in the schema
$(\BSig{0} \comp)$, one can only deduce the existence of a set with a
bounded domain.
\begin{example}
 \label{ex:CA}
 The axiom $(\BSig{0} \comp)$ of comprehension allows us to transform
 given sets $\vec X$ into another set $Y$ via $\BSig{0}$-definable
 operations without inessential encodings.
 For an easy example, assume that two sets $U$ and $V$ encode binary strings respectively of
 length $m$ and $n$ in such a way that
 $j \in U$ $\Leftrightarrow$ ``the $j^{\mathrm{th}}$ bit of the string $U$ is $1$'' and
 $j \not\in U$ $\Leftrightarrow$ ``the $j^{\mathrm{th}}$ bit of the string $U$ is $0$''
 for each $j < m$.
 Then the {\em concatenation} $W = U \CON V$, the string $U$ followed by $V$, is
 defined by $(\BSig{0} \comp)$ as follows.
 \[
 (\forall j < m+n)
 \big[ j \in W \leftrightarrow
       \big( (j < m  \wedge j \in U) \vee
              (m \leq j \wedge j-m \in V)
       \big)
 \big]
 \]
\end{example}
\begin{definition}[Definable functions in formal systems]\label{d:ptf}
Let $T$ be one of the formal systems defined above and
$\Phi$ be a class of bounded formulas.
A function 
$f: \mathbb N^k \rightarrow \mathbb N$
is {\em $\Phi$-definable in $T$} if there exists a formula
$\varphi (x_1, \dots, x_k, y) \in \Phi$
with no other free variables such that
$\varphi (\vec x, y)$ expresses the relation
$f(\vec x) = y$
(under the standard semantics) and $T$ proves the sentence
$\forall \vec x ~ \exists ! y \varphi (\vec x, y)$.
\end{definition}
\begin{theorem}[\cite{buss86}]
\label{t:Buss86}%
\begin{enumerate}%
\item A function is $\bSig{1}$-definable in $\mS^1_2$ if and only if it is
 computable in polynomial time.
\label{t:Buss86:S}
\item A function is $\BSig{1}$-definable in $\mU^1_2$ if and only
  if it is computable in polynomial space.
\label{t:Buss86:U}
\end{enumerate}  
\end{theorem}
To readers who are not familiar with second order bounded arithmetic,
it might be of interest to outline the proof that every polynomial-space
computable function can be defined in $\mU^1_2$.
The argument is commonly known as the
{\em divide-and-conquer} method, which was originally used to show the
classical inclusion
$\mathrm{NPSPACE} \subseteq \mathrm{PSPACE}$ \cite{Savitch70}.
 \begin{proof}[Proof of the ``if'' direction of Theorem
  \ref{t:Buss86}.\ref{t:Buss86:U} (Outline)]
 Suppose that a function
 $f: \mathbb N^k \rightarrow \mathbb N$
 is computable in polynomial space.
  This means that there exist
   a deterministic Turing machine $M$ and a polynomial
 $p: \mathbb N^k \rightarrow \mathbb N$
 such that, for any inputs $m_1, \dots, m_k$,
 $f(m_1, \dots, m_k)$
 can be computed by $M$ while the head of $M$ only visits a number of cells bounded by
 $p(|m_1|, \dots, |m_k|)$.
 Then, since the number of possible configurations under $M$ on inputs
 $m_1, \dots, m_k$ is bounded by
 $2^{q (|m_1|, \dots, |m_k|)}$ for some polynomial $q$,
 the computation terminates in a step
 bounded by $2^{q (|\vec m|)}$ as well.

 Let
 $\psi_M (m_1, \dots, m_k, n, w_0, W)$
 denote a $\BSig{0}$-formula expressing that the set
 $W$ encodes the concatenation 
 $w_1 \CON \cdots \CON w_{2^{|n|}}$
 of configurations under $M$, where $w_j$ is the next configuration of
 $w_{j-1}$, writing $w_j = \m{Next}_M (w_{j-1})$
 $(1 \leq j \leq 2^{|n|})$.
 Reasoning informally in $\mU^1_2$,
 the $\BSig{1}$-formula
 $\varphi (\vec m, n) : \equiv
 (\forall w \leq 2^{p (|\vec m|)}) \exists W
 \psi_M (\vec m, n, w, W)$
  can be deduced by $(\BSig{1} \pind)$ on $n$.
  In case $n=0$, $W$ can be defined identical to
  $\m{Next}_M (w)$.
  For the induction step, given a configuration
  $w_0 \leq 2^{p (|\vec m|)}$,
  the induction hypothesis yields a set $U$ such that
  $\psi_M (\vec m, \divtwo{n}, w_0, U)$ holds.
  Another instance of the induction hypothesis yields a set
  $V$ such that
  $\psi_M (\vec m, \divtwo{n}, w_{2^{|n| -1}}, V
          )$
  holds.
  Since $2^{|n|} = 2^{|n| -1} + 2^{|n| -1}$,
  $\psi_M (\vec m, n, w_0, W)$ holds for the set
  $W := U \CON V$.

  Now instantiating $n$ with $2^{q (|\vec m|)}$ yields a set
  $W$ such that
  $\psi_M (\vec m, 2^{q (|\vec m|)}, \m{Init}_M (\vec m), W)$ holds
  for the initial configuration $\m{Init}_M (\vec m)$ on inputs $\vec
  m$.
  The set $W$ yields the final configuration and thus the
  result $f(\vec m)$ of the computation.
  The uniqueness of the result can be deduced in
  $\mU^1_2$ accordingly.
 \end{proof}
The ``only if'' direction of Theorem \ref{t:Buss86}.\ref{t:Buss86:U} follows from a bit more general statement.
\begin{lemma}
 \label{l:Buss}
 If $\mU^1_2$ proves
 $\exists y \varphi (x_1, \dots, x_k, y)$ for a
 $\BSig{1}$-formula $\varphi (x_1, \dots, x_k, y)$ with no other
 free variables, then there exists a
 function
 $f: \mathbb N^k \rightarrow \mathbb N$
 such that, for any naturals
 $\vec m = m_1, \dots, m_k \in \mathbb N$, 
{\normalfont (i)} $f(\vec m)$ is computable with the use of space bounded by a polynomial in
 $|m_1|, \dots, |m_k|$, and
 {\normalfont (ii)}
 $\varphi (\underline{\vec m}, \underline{f(\vec m)})$ holds under the
 standard semantics, where $\underline{m}$ denotes
 the numeral $\mrS^m (0)$ for a natural $m$.
\end{lemma}
 
It is also known that the second order system axiomatized with the schema
$(\BSig{1} \ind)$, instead of $(\BSig{1} \pind)$, of 
the usual induction
$\varphi (0) \wedge
 \forall x \left( \varphi (x) \rightarrow  \varphi (\mrS (x))
           \right)
 \rightarrow \forall x \varphi (x)$
for $\BSig{1}$-formulas, called $\mV^1_2$, captures the exponential-time computable
functions of polynomial growth rate in the sense of Theorem
\ref{t:Buss86}.
Though there is no common notion about what is bounded arithmetic, the
exponential function $m \mapsto 2^m$ is not definable in any existing
system of bounded arithmetic.

%% file: mfg.tex
\section{Minimal function graphs}
\label{s:mfg}

The {\em minimal function graph} semantics was described in
\cite{JonesM86} as denotational semantics,
cf. \cite[Chapter 9]{Winskel93}, and afterward used for termination analysis of
functional programs without exponential size-explosions
in \cite[Chapter 24.2]{Jones97} and \cite{Marion03}.
In this section, we explain how
minimal function graphs work, how they are defined inductively,
and how they can be defined without inductive definitions.

To see how minimal function graphs work, consider the program
$\RS_{\m{lcs}}$ in Example \ref{ex:LCS}.
Let us observe that the following reduction starting with the basic term
$\m{lcs} (\m{a} (\m{a} (\epsilon)),
         \m{b} (\m{b} (\epsilon))
     )
$
is possible.
\begin{equation*}
\begin{array}{rl}
 & \m{lcs} (\m{a} (\m{a} (\epsilon)),
         \m{b} (\m{b} (\epsilon))
	 ) \\
 \rewrite{\mathsf i}{}{\RS_{\m{lcs}}} &
   \m{max} (\m{lcs} (\m{a} (\epsilon), \m{b} (\m{b} (\epsilon))),
        \m{lcs} (\m{a} (\m{a} (\epsilon)), \m{b} (\epsilon))
   ) \\   
 \rewrite{\mathsf i}{}{\RS_{\m{lcs}}} &
   \m{max} (\m{lcs} (\m{a} (\epsilon), \m{b} (\m{b} (\epsilon))),
       \m{max}(	    
            \m{lcs} (\m{a} (\epsilon), \m{b} (\epsilon)),
            \m{lcs} (\m{a} (\m{a} (\epsilon)), \epsilon)
	  )  
   ) \\   
 \rewrite{\mathsf i}{}{\RS_{\m{lcs}}} &
  \m{max} (\m{max}(
            \m{lcs} (\epsilon, \m{b} (\m{b} (\epsilon))),
            \m{lcs} (\m{a} (\epsilon), \m{b} (\epsilon))
          ),
       \m{max}(	    
            \m{lcs} (\m{a} (\epsilon), \m{b} (\epsilon)),
            \m{lcs} (\m{a} (\m{a} (\epsilon)), \epsilon)
	  )  
      )
\end{array} 
\end{equation*}	 
In the reduction, the term
$t := \m{lcs} (\m{a} (\epsilon), \m{b} (\epsilon))$
is duplicated, and hence costly re-computations potentially occur.
For the same reason, there can be an exponential explosion in the size
of the reduction tree rooted at
$\m{lcs} (\m{a} (\m{a} (\epsilon)),
         \m{b} (\m{b} (\epsilon))
     )
$
that contains all the possible
rewriting sequences starting with the basic term.
A minimal function graph $G$, or {\em cache} in other words, is defined so that $G$ stores pairs of a basic term and its
normal form.
Thus, once the term $t$ is normalized to $\m{0}$
(because the two strings $\m{a}$ and $\m{b}$ have no common subsequence),
the pair $\seq{t, 0}$ is stored in $G$ and
any other reduction of $t$ can be simulated by replacing the
occurrence of $t$ with $\m{0}$.

Given a program $\RS$, a (variant of) minimal function graph $G$ is defined as the least
fixed point of the following operator $\Gamma$ over
$\mathcal P (\BF{} \times \TC{})$,
where $X \subseteq \BF{} \times \TC{}$.
\begin{equation*}
 \begin{array}{rcl}
 \seq{t, s} \in \Gamma (X) & :\Longleftrightarrow &
 \exists l \rightarrow r \in \RS,
 \exists \theta: \VS \rightarrow \TC{},
 \exists \seq{t_0, s_0}, \dots,
         \seq{t_{\size{r} -1}, s_{\size{r} -1}} \in X \text{ s.t.} \\
 &&
 t = l \theta \ \& \
 s = \big( (r \theta)[s_0/t_0] \cdots \big) [s_{\size{r} -1} / t_{\size{r} -1}]
 \end{array}
\end{equation*}
The operator $\Gamma$ is monotone, i.e.,
$X \subseteq Y \Rightarrow \Gamma (X) \subseteq \Gamma (Y)$, and hence
there exists the least fixed point of $\Gamma$.
Suppose that $\RS$ is quasi-reducible.
On one side, the fixed-ness of $G$ yields that
$t \rewnf s \Rightarrow \seq{t, s} \in G$.
On the other side, since the set
 $\{ \seq{t, s} \mid t \in \BF{} \ \& \ t \rewnf s \}$
 is a fixed point of $\Gamma$,
 the least-ness of $G$ yields that
 $\seq{t, s} \in G \Rightarrow t \rewnf s$.
 Thus, to conclude that every closed basic term has an (innermost)
 $\RS$-normal form, it suffices to show that, for every term
 $t \in \BF{}$, there exists a term $s$ such that
 $\seq{t, s} \in G$.
Now there are two important observations.
\begin{enumerate}
 \item It suffices to show that, for every term $t \in \BF{}$, there exist
a {\em subset} $G_t \subseteq G$ and a term $s$ such that $\seq{t, s} \in
       G_t$.
       If $t = l \theta$ and
       $s = \big( (r \theta)[s_0/t_0] \cdots \big) [s_{\size{r} -1} / t_{\size{r} -1}]$
       as in the definition of $\Gamma$ above and, for each $j < \size{r}$,
       $\seq{t_j , s_j} \in G_{t_j}$ holds for such a set
       $G_{t_j} \subseteq G$, then
       $G_t$ can be simply defined as
       $G_t = \{ \seq{t, s} \} \cup 
              G_{t_0} \cup \cdots \cup G_{t_{\size{r} -1}}
       $.%
       \footnote{To be precise, in \cite{Jones97,Marion03}, the {\em minimal function
       graph} was used to denote such a subset $G_t$ for a given basic $t$.}
 \item Additionally suppose that the program $\RS$ reduces under an LPO $<_{\lpo}$.
       Then it turns out that the definition of $\Gamma$ is equivalent to a
       form restricted in such a way that
       $t_j <_{\lpo} t$ for each $j < \size{r}$.%
       \footnote{Namely, every function computed by an $<_{\lpo}$-reducing
       program is defined recursively along $<_{\lpo}$.
       Therefore, as a reviewer pointed out, in this case the minimal function graphs can
       be regarded as fixed-point semantics for recursive
       definitions of functions, cf. \cite[Chapter 10]{SlonnegerK95}.}\label{rem:mfg}
\end{enumerate} 
For these reasons, the schema
$(\forall t \in \BF{})
 \left( (\forall s <_{\lpo} t) \varphi (s) \rightarrow \varphi (t)
 \right) \rightarrow
 (\forall t \in \BF{}) \varphi (t)$       
of transfinite induction along $<_{\lpo}$ will imply the termination of a
quasi-reducible LPO-terminating program $\RS$ in the sense above.

%% file: termination.tex
\section{Formalizing LPO-termination proofs under PQIs in $\mU^1_2$}
\label{s:termination}

In this section, we show that, if 
$\RS$ is a quasi-reducible $\LPOPoly$-program, then an innermost
$\RS$-normal form of any closed
basic term can be found in the system
$\mU^1_2$ (Theorem \ref{t:MFG}).

Given a program $\RS$ over a signature
$\FS = \CS \cup \DS$, we use the notation $V_{\RS}$ to denote the finite set
$\{ x \in \VS \mid
    x \text{ appears in some rule } \rho \in \RS
 \}$
of variables.
Let $\code{\cdot}$ be an {\em efficient} binary encoding for
$\mathbf T (\FS, \VS_{\RS})$-terms.
\label{p:code}
The efficiency means that:
\begin{enumerate}
\renewcommand{\theenumi}{(\roman{enumi})}
\renewcommand{\labelenumi}{(\roman{enumi})}
 \item $t \mapsto \code{t}$ is $\BSig{0}$-definable in $\mU^1_2$.\label{en:code:1}
 \item There exists a polynomial (term) $p(x)$ with a free variable $x$
       such that $|\code{t}| \leq p(\size{t})$ (provably) holds for any
       $t \in \mathbf T (\FS, \VS_{\RS})$.\label{en:code:2}
 \end{enumerate}
Without loss of generality, we can assume that:
 \begin{enumerate}
\renewcommand{\theenumi}{(\roman{enumi})}
\renewcommand{\labelenumi}{(\roman{enumi})}
  \setcounter{enumi}{2}
  \item $\size{t} \leq |\code{t}|$.\label{en:code:3}
  \item $|\code{s}| < |\code{t}|$ if
$s$ is a proper subterm of $t$.\label{en:code:4} 
 \end{enumerate}
Such an encoding can be defined, for example, by representing terms as
directed graphs not as trees.

\begin{lemma}
\label{l:lpoU}  
The relation $<_{\lpo}$ is 
$\BSig{0}$-definable in $\mU^1_2$.
 \end{lemma}

%
\begin{proof}[Proof (Sketch)]
It suffices to show that, given two terms $s$ and $t$, the relation
``there exists a derivation tree according to the rules
\ref{d:lpo:1}--\ref{d:lpo:3} (on page \pageref{page:lpo}) that results in
$s <_{\lpo} t$'' is
$\BSig{0}$-definable in $\mU^1_2$.
Let $T$ denote such a derivation tree resulting in
$s <_{\lpo} t$.
By induction according to the inductive definition of $<_{\lpo}$
it can be shown that the number of nodes in $T$ is bounded by
$\size{s} \cdot \size{t}$.
Hence, by the assumption \ref{en:code:2} on the encoding $\code{\cdot}$,
the code $\code{T}$ of $T$ is polynomially bounded in
$\size{s} \cdot \size{t}$ and thus in
$\code{s} \cdot \code{t}$.
On the other hand, by definition, the relation
$s_0 <_{\lpo} t_0$ between two terms 
$s_0$ and $t_0$ is reduced to a tuple
$s_j <_{\lpo} t_j $ $(j = 1, \dots, k)$
of relations between some subterms
$s_1, \dots, s_k$ of $s_0$ and subterms
$t_1, \dots, t_k$ of $t_0$.
Thanks to the assumption \ref{en:code:4} on the encoding $\code{\cdot}$,
$|\code{s_j}| + |\code{t_j}| < |\code{s_0}| + |\code{t_0}|$,
i.e.,
 $2^{|\code{s_j}| + |\code{t_j}|} \leq
  \divtwo{\left( 2^{|\code{s_0}| + |\code{t_0}|} \right)}$,
holds for any $j \in \{ 1, \dots, k \}$.
From these observations, it can be seen that the construction of the
 derivation tree $T$ is performed in $\mU^1_2$, and hence the
 relation $s <_{\lpo} t$ is
$\BSig{0}$-definable in $\mU^1_2$.
\end{proof}

As observed in \cite{buch95},
in which an optimal LPO-termination proof was described,
every program $\RS$ reducing under an LPO $<_{\lpo}$
already reduces under a finite
restriction $<_{\ell}$ of $<_{\lpo}$ for some
$\ell \in \mathbb N$ and
every quantifier of the form $(Q s <_{\ell} t)$ can be regarded as a
bounded one.
Adopting the restriction, we introduce an even more restrictive relation
$<_{\ell}$ ($\ell \in \mathbb N$)
motivated by the following properties of PQIs. 
\begin{proposition}
 \label{p:PQI}
 Let $\qint{\cdot}$ be a kind $0$ PQI and
 $t \in \BF{}$.
 Then the following two properties hold.
 \begin{enumerate}
  \item $\qint{t} \leq p (|\code{t}|)$ holds for some polynomial $p$.\label{p:PQI:1}
  \item Suppose additionally that a program $\RS$ admits the PQI
	$\qint{\cdot}$ and that
	$t \rewast s$ holds.
        If $s \in \TC{}$, then
        $\size{s} \leq \qint{t}$ holds.
        If $s = f(s_1, \dots, s_k) \in \BF{}$, then
        $\size{s_j} \leq \qint{t}$
        holds for each $j \in \{ 1, \dots, k \}$.\label{p:PQI:2}
 \end{enumerate}
\end{proposition}
 \begin{proof}
  {\sc Property} \ref{p:PQI:1}.
  Let $t = g(t_1, \dots, t_l)$.
  Since the PQI $\qint{\cdot}$ is kind $0$,
  one can find a constant $d$ depending only on the set
  $\CS$ of constructors and the PQI $\qint{\cdot}$ such that
  $\qint{t_j} \leq d \cdot \size{t_j}$
  holds
  for any $j \in \{ 1, \dots, l \}$.
  This yields a polynomial $p$ such that
  $\qint{t} \leq p (\size{t})$ and thus
  $\qint{t} \leq p(|\code{t}|)$ holds by the assumption \ref{en:code:3} on the encoding $\code{\cdot}$.

  {\sc Property} \ref{p:PQI:2}.
  In case $s \in \TC{}$,
  $\size{s} \leq \qint{s} \leq \qint{t}$ holds.
  In case $s = f(s_1, \dots, s_k) \in \BF{}$,
  $\size{s_j} \leq \qint{s_j} \leq \qint{s} \leq \qint{t}$
  holds for each $j \in \{ 1, \dots, k \}$.
 \end{proof}
\begin{definition}[$\TC{\ell}$, $\BF{\ell}$, $<_{\ell}$, $\lex{<_{\ell}}$]
\label{d:lpoell}  
Let
$\TC{\ell}$ denote a set
$\{ t \in \TC{} \mid \size{t} \leq \ell \}$
of constructor terms and $\BF{\ell}$ a set
$\{ f(t_1, \dots, t_k) \in \BF{} \mid 
    \size{t_1}, \dots, \size{t_k} \leq \ell
 \}$
of basic terms.
Then we write
$s <_{\ell} t$ if
$s <_{\lpo} t$ and additionally
$s \in \TC{\ell} \cup \BF{\ell}$ hold.
We use the notation
$s <_{\ell}^{\!\text{\tiny $\langle i \rangle$}} t$
($i = \ref{d:lpo:1}, \ref{d:lpo:2}, \ref{d:lpo:3}$)
 accordingly.
Moreover, we define a
 {\em lexicographic extension}
 $\lex{<_{\ell}}$ of $<_{\ell}$ over $\TC{}$.
For constructor terms $s_1, \dots, s_k$, $t_1, \dots, t_k$,
 we write
 $(s_1, \dots, s_k) \lex{<_{\ell}} (t_1, \dots, t_k)$
 if there exists an index $i \in \{ 1, \dots, k \}$ such that
 $s_j = t_j$ for every $j < i$, $s_i \dlpo{\ell}{1} t_i$, and
 $s_j \in \TC{\ell}$ for every $j > i$.
\end{definition}
Corollary \ref{c:lex} follows from the definitions of $<_{\ell}$ and
$\lex{<_{\ell}}$ and from $<_{\FS}$-minimality of constructors.
\begin{corollary}
 \label{c:lex}
 For two basic terms
 $f(s_1, \dots, s_k), f(t_1, \dots, t_k) \in \BF{\ell}$,
 $f(s_1, \dots, s_k) \dlpo{\ell}{3} f(t_1, \dots, t_k)$
 holds if and only if
 $(s_1, \dots, s_k) \lex{<_{\ell}} (t_1, \dots, t_k)$
 holds.
\end{corollary}
For most of interesting $\LPOPoly$-programs including Example \ref{ex:LCS} and \ref{ex:QBF}, interpreting polynomials consist of
$+$, $\cdot$, $\max_{j=1}^k x_j$ together with additional constants.
This motivates us to formalize PQIs limiting interpreting polynomial
terms to those built up only from $0$, $\mrS$, $+$, $\cdot$ and $\max$
to make the formalization easier.
Then the constraints (ii) and (iii) on PQIs follow from defining axioms
for these function symbols.

Let us consider a reduction $t_0 \rewast t \rewast s$ under a program
$\RS$ admitting a kind $0$ PQI $\qint{\cdot}$, where
$t_0, t \in \BF{}$ and $s \in \TC{} \cup \BF{}$.
If $s <_{\lpo} t$ for some LPO $<_{\lpo}$, then
Proposition \ref{p:PQI} yields a polynomial $p$ such that  
$s <_{p (|\code{t_0}|)} t$ holds by Definition \ref{d:lpoell}.
Hence we can assume that $\ell$ is (the result of substituting $t_0$ for) a polynomial
$p(|x|)$.
More precisely, $\ell$ can be expressed by an 
$\LBA$-term built up from $0$ and $|x|, |y|, |z|, \dots$ by
$\mathrm S$, $+$ and ${}\cdot{}$.
By assumption, $\ell$ does not contain $\#$ nor $\divtwo{\cdot}$.
Thus $\ell = \ell (x_1, \dots, x_k)$ denotes a polynomial with
 non-negative coefficients in
$|x_1|, \dots, |x_k|$. 
Since $\ell$ contains no smash $\#$ in particular, $2^{p(\ell)}$ can be
 regarded as an $\LBA$-term for any polynomial $p(x)$.
By the assumption \ref{en:code:2} on the encoding $\code{\cdot}$,
$|\code{t}|$ is polynomially bounded in
the size $\size{t}$ of $t$, and hence
$\code{t} \leq 2^{p(\size{t})}$ for some polynomial $p(x)$.
Therefore any quantifier of the forms
$(Q s <_{\ell} t)$, $(Q t \in \TC{\ell})$ and 
$(Q t \in \BF{\ell})$ can be treated as a bounded one.

We deduce the schema
$\eqref{e:TI:B}$ of $<_{\ell}$-transfinite induction over $\BF{\ell}$ for
 $\BSig{1}$-formulas (Lemma \ref{l:TI:B}).
 Since the relation
 $f(s_1, \dots, s_k) \dlpo{\ell}{3} f(t_1, \dots, t_k)$
 relies on the comparison
 $(s_1, \dots, s_k) \lex{<_{\ell}}
 (t_1, \dots, t_k)$
 by Corollary \ref{c:lex},
  we previously have to deduce the schema
  $\eqref{e:TI:k}$ of $\lex{<_{\ell}}$-transfinite induction over
  $k$-tuples of $\TC{\ell}$-terms (Lemma \ref{l:TI:k}).
We start with deducing the instance in the base case $k = 1$.

\begin{lemma}
\label{l:TI:C}  
The following schema of $<_{\ell}$-transfinite induction over
$\TC{\ell}$ holds in $\mU^1_2$, where
$\varphi \in \BSig{1}$.
\begin{equation}
\tag{$\m{TI}_{\BSig{1}} \left( \TC{\ell}, <_{\ell} \right)$}
\label{e:TI}
(\forall t \in \TC{\ell})
\big( (\forall s <_{\ell} t) \varphi (s)
      \rightarrow \varphi (t)
\big)
\rightarrow
(\forall t \in \TC{\ell}) \varphi (t)
\quad
\end{equation}
\end{lemma}

\begin{proof}
Reason in $\mU^1_2$.
Suppose 
$(\forall t \in \TC{\ell})
\big( (\forall s <_{\ell} t) \varphi (s)
      \rightarrow \varphi (t)
\big)
$
and let 
$t \in \TC{\ell}$.
We show that $\varphi (t)$ holds by
$(\BSig{1} \pind)$ on $\code{t}$.
The case $\code{t} = 0$ trivially holds.
Suppose $\code{t} >0$ for induction step.
By assumption, it suffices to show that
$\varphi (s)$ holds for any $s <_{\ell} t$.
Thus let $s <_{\ell} t$.
 Since $t \in \TC{\ell}$, $s$ is a proper subterm of $t$ by Corollary
 \ref{c:lpo} and $<_{\FS}$-minimality of constructors.
Thus, the assumption \ref{en:code:4} on the encoding $\code{\cdot}$ yields
$\code{s} \leq \divtwo{\code{t}}$, and hence
$\varphi (s)$ holds by induction hypothesis.
\end{proof}

\begin{remark}
 In the proof of Lemma \ref{l:TI:C}, we employed a bit-wise form of
 {\em course of values} induction
$\varphi (0) \wedge
 \forall t
 \big( \forall s (\code{s} \leq \divtwo{\code{t}} \rightarrow
                  \varphi (s)
                 )
       \rightarrow \varphi (t) 
 \big)
 \rightarrow \forall t \varphi (t)$
 for a $\BSig{1}$-formula $\varphi (x)$, which is not an instance of
 $(\BSig{1} \pind)$.
 Formally, one should apply $(\BSig{1} \pind)$ for the $\BSig{1}$-formula
 $\psi (x) \equiv \forall t
 \big( \code{t} \leq 2^{|x|} \rightarrow \varphi (t)
 \big)$
 to deduce
 $\left( \forall t \in \TC{\ell} \right) \varphi (t)$.
To ease presentation, we will use similar informal arguments in the sequel.
\end{remark}
   
\begin{lemma}
\label{l:TI:k}
The schema {\normalfont (\ref{e:TI})} can be extended to tuples of
$\TC{\ell}$-terms, i.e., 
the following schema holds in $\mU^1_2$, where
$\varphi (\vec t) \equiv
 \varphi (t_1, \dots, t_k) \in \BSig{1}$.
\begin{equation}
\tag{$\m{TI}_{\BSig{1}} ( \TC{\ell}^k, \lex{<_{\ell}} )$}
\label{e:TI:k}
(\forall \vec t \in \TC{\ell})
\big( (\forall \vec s \lex{<_{\ell}} \vec t) \varphi (\vec s)
      \rightarrow \varphi (\vec t)
\big)
\rightarrow
(\forall \vec t \in \TC{\ell}) \varphi (\vec t)
\end{equation}
\end{lemma}

\begin{proof}
We show that the schema (\ref{e:TI:k}) holds in 
$\mU^1_2$ by (meta) induction on $k \geq 1$.
In case $k=1$, the schema is an instance of 
(\ref{e:TI}).
Suppose that $k > 1$ and 
($\m{TI}_{\BSig{1}} ( \TC{\ell}^{k-1}, \lex{<_{\ell}} )$)
holds by induction hypothesis.
Assume that
\begin{equation}
\label{e:prog}
(\forall t_1, \dots, t_k \in \TC{\ell})
\big( (\forall (s_1, \dots, s_k) \lex{<_{\ell}} (t_1, \dots, t_k)
      ) \varphi (s_1, \dots, s_k)
\rightarrow \varphi (t_1, \dots, t_k) 
\big)
\end{equation}
holds for some $\BSig{1}$-formula
$\varphi (t_1, \dots, t_k)$.
Let
$\varphi_{\lex{<_{\ell}}} (t, t_2, \dots, t_k)$,
$\psi (t)$ and
$\psi_{<_{\ell}} (t)$
denote $\BSig{1}$-formulas specified as follows.
\begin{equation*}
\begin{array}{rcl}
\varphi_{\lex{<_{\ell}}} (t, t_2, \dots, t_k) &:\equiv&
t_2, \dots, t_k \in \TC{\ell} \wedge 
\left( \forall (s_2, \dots, s_k) \lex{<_{\ell}}
       (t_2, \dots, t_k) 
\right) \varphi (t, s_2, \dots, s_k); \\
\psi (t) &:\equiv&
\left( \forall t_2, \dots, t_k \in \TC{\ell} \right)
\varphi (t, t_2, \dots, t_k); \\
\psi_{<_{\ell}} (t) &:\equiv&
t \in \TC{\ell} \wedge (\forall s <_{\ell} t) \psi (s).
\end{array}
\end{equation*}
Note, in particular, that
$\psi (t)$ is still a $\BSig{1}$-formula
since every quantifier of the form
$(\forall s \in \TC{\ell})$
can be regarded as a bounded one under which the class
$\BSig{1}$ is closed. 
One can see that
$\varphi_{\lex{<_{\ell}}} (t, t_2, \dots, t_k)$ and
$\psi_{<_{\ell}} (t)$
imply
$t, t_2, \dots, t_k \in \TC{\ell}$
and
$\big( \forall (s, s_2, \dots, s_k) \lex{<_{\ell}}
        (t, t_2, \dots, t_k)
 \big)
        \varphi (s, s_2, \dots, s_k)
$.
Hence, by the assumption \eqref{e:prog},
$\psi_{<_{\ell}} (t)$
implies
$(\forall t_2, \dots, t_k \in \TC{\ell})
 \big( \varphi_{\lex{<_{\ell}}} (t, t_2, \dots, t_k) \rightarrow
       \varphi (t, t_2, \dots, t_k)
 \big)$,
which denotes
\begin{equation*}
\label{e:prog:k-1}
(\forall t_2, \dots, t_k \in \TC{\ell})
\big( (\forall (s_2, \dots, s_k) \lex{<_{\ell}} (t_2, \dots, t_k)
      ) \varphi (t, s_1, \dots, s_k)
\rightarrow \varphi (t, t_2, \dots, t_k) 
\big).
\end{equation*}
This together with
($\m{TI}_{\BSig{1}} ( \TC{\ell}^{k-1}, \lex{<_{\ell}} )$)
yields
$\left( \forall t_2, \dots, t_k \in \TC{\ell} \right)
 \varphi (t, t_2, \dots, t_k)$,
denoting 
$\psi (t)$.
This means that
$(\forall t \in \TC{\ell})
 \big( (\forall s <_{\ell} t) \psi (s) \rightarrow \psi (t)
 \big)$ 
holds.
Since $\psi (t) \in \BSig{1}$ as noted above,
this together with (\ref{e:TI}) yields
$(\forall t \in \TC{\ell}) \psi (t)$
and thus
$(\forall t_1, \dots, t_k \in \TC{\ell})
 \varphi (t_1, \dots, t_k)$
holds.
\end{proof}

\begin{lemma}
\label{l:TI:B}
Let $\FS = \CS \cup \DS$.
The $<_{\ell}$-transfinite induction over
$\BF{\ell}$ holds in $\mU^1_2$, where
$\varphi \in \BSig{1}$.
\begin{equation}
\tag{$\m{TI}_{\BSig{1}} \left( \BF{\ell}, <_{\ell} \right)$}
\label{e:TI:B}
(\forall t \in \BF{\ell})
\big( (\forall s \in \BF{\ell})
      (s <_{\ell} t \rightarrow \varphi (s))
      \rightarrow \varphi (t)
\big)
\rightarrow
(\forall t \in \BF{\ell}) \varphi (t)
\quad
\end{equation}
\end{lemma}

Given a precedence $<_{\FS}$ on the finite signature $\FS$, let
$\rk: \FS \rightarrow \mathbb N$
denote the {\em rank}, a finite function compatible with $<_{\FS}$:
$\rk (f) < \rk (g) \Leftrightarrow f <_{\FS} g$.

\begin{proof}
Reason in $\mU^1_2$.
Assume the premise of $\eqref{e:TI:B}$:
\begin{equation}
\label{e:l:prog}
(\forall t \in \BF{\ell})
\big( (\forall s \in \BF{\ell})
      (s <_{\ell} t \rightarrow \varphi (s))
      \rightarrow \varphi (t)
\big)
\end{equation}
Let $g \in \DS$.
We show that
$(\forall t_1, \dots, t_l \in \TC{\ell})
 \varphi (g (t_1, \dots, t_l))$
holds by $(\BSig{1} \pind)$ on 
$2^{\rk (g)}$, or in other words by finitary induction on $\rk (g)$.
Let $t_1, \dots, t_l \in \TC{\ell}$ and
$t := g(t_1, \dots, t_l)$.
By the assumption \eqref{e:l:prog}, it suffices to show that
$\varphi (s)$ holds for any
$s \in \BF{\ell}$ such that $s <_{\ell} t$.
Thus, let $s \in \BF{\ell}$ and $s <_{\ell} t$.

{\sc Case.} 
$s \dlpo{\ell}{1} t$:
In this case
$s \leqslant_{\ell} t_i$ for some $i \in \{ 1, \dots, l \}$.
Since $t_i \in \TC{\ell}$, $s \in \TC{\ell}$ as well by Corollary
\ref{c:lpo}, and hence this case is excluded.

{\sc Case.}
$s := f(s_1, \dots, s_k) \dlpo{\ell}{2} t$:
In this case, $f <_{\FS} g$ and hence $\rk (f) < \rk (g)$.
This allows us to reason as
$2^{\rk (g)} \leq
 2^{\rk (f) -1} =
 \divtwo{2^{\rk (f)}}$.
Thus the induction hypothesis yields $\varphi (s)$.

{\sc Case.}
$s := g(s_1, \dots, s_l) \dlpo{\ell}{3} t$:
We show that the following condition holds.
\begin{equation}
\label{e:l:prog:l}
(\forall v_1, \dots, v_l \in \TC{\ell})
\big( (\forall (u_1, \dots, u_l) \lex{<_{\ell}} (v_1, \dots, v_l)
      ) \varphi (g(u_1, \dots, u_l))
\rightarrow \varphi (g(v_1, \dots, v_l)) 
\big)
\end{equation}
Let $v_1, \dots, v_l \in \TC{\ell}$.
By Corollary \ref{c:lex}, the premise
$(\forall (u_1, \dots, u_l) \lex{<_{\ell}} (v_1, \dots, v_l))
 \varphi (g(u_1, \dots, u_l)$
of \eqref{e:l:prog:l} yields
$\left( \forall s'  \dlpo{\ell}{3} g(v_1, \dots, v_l)
 \right) \varphi (s')
$.
On the other side, the previous two cases yield
$(\forall s' \in \BF{\ell})
 \big(s' <_{\ell}^{\!\text{\tiny $\langle i \rangle$}}
      g(v_1, \dots, v_l)
 \rightarrow \varphi (s')
 \big)$
$(i = \ref{d:lpo:1}, \ref{d:lpo:2})$
and hence
$(\forall s' \in \BF{\ell})
 \big(s' <_{\ell} g(v_1, \dots, v_l)
 \rightarrow \varphi (s')
 \big)$
holds. 
Therefore $\varphi (g(v_1, \dots, v_l))$ holds by the assumption
\eqref{e:l:prog}, yielding the statement \eqref{e:l:prog:l}.
Since \eqref{e:l:prog:l} is the premise of an instance of the schema
$(\m{TI}_{\BSig{1}} ( \TC{\ell}^{l}, \lex{<_{\ell}} ))$,
Lemma \ref{l:TI:k} yields
$(\forall v_1, \dots, v_l \in \TC{\ell})
 \varphi (g(v_1, \dots, v_l))$,
and thus
$\varphi (g(s_1, \dots, s_l))$ holds in particular.
\end{proof}

To derive, from
$\eqref{e:TI:B}$, the existence of a minimal function graph under an
LPO-terminating program, we need the following technical lemma.
\begin{lemma}
\label{l:sub}
{\normalfont (in $\mU^1_2$)}
Let $\qint{\cdot}$ be a kind $0$ PQI for a signature
$\FS = \CS \cup \DS$,
$t \in \mathbf{B(F)}$, $s \in \mathbf{T(F)}$
and $<_{\lpo}$ an LPO induced by a precedence 
$<_{\FS}$.
If
$s <_{\lpo} t$ and
$\qint{s} \leq \qint{t} \leq \ell$,
then, for any basic subterm $t'$ of $s$ and for any
$s' \in \mathbf{T(C)}$ such that
$\qint{s'} \leq \qint{t'}$,
$v <_{\ell} t$
holds for any basic subterm $v$ of $s[s'/t']$.
\end{lemma}

\begin{proof}
By $<_{\FS}$-minimality of constructors,
$s ' <_{\lpo} t'$ holds.
Hence $s[s'/t'] <_{\lpo} t$ from the assumption
$s <_{\lpo} t$.
This yields $v <_{\lpo} t$ by the definition of LPOs.
Write $v = f(v_1, \dots, v_k)$ for some $f \in \DS$ and
$v_1, \dots, v_k \in \mathbf{T(C)}$.
Let $i \in \{ 1, \dots, k \}$.
Then
$\size{v_i} \leq \qint{v_i} \leq \qint{v} \leq
 \qint{s[s'/t']} \leq \qint{t}$.
The last inequality follows from the monotonicity (ii) of the PQI
$\qint{\cdot}$.
This yields $\size{v_i} \leq \ell$ and hence $v <_{\ell} t$.
\end{proof}

\begin{theorem}
\label{t:MFG}
{\normalfont (in $\mU^1_2$)}
Suppose that $\RS$ is a quasi-reducible
$\LPOPoly$-program.
Then, for any basic term $t$, there exists a
 minimal function graph $G$
 (in the sense of Section \ref{s:mfg}) such that that
 $\seq{t, s} \in G$ holds for an $\RS$-normal form
 $s$ of $t$.
\end{theorem}

\begin{proof}
Suppose that $\RS$ is a quasi-reducible
$\LPOPoly$-program witnessed by an LPO
$<_{\lpo}$ and a kind $0$ PQI $\qint{\cdot}$ and that $<_{\ell}$ is a
finite restriction of $<_{\lpo}$.
Let $\psi_{\ell} (x, y, X)$ denote a $\BSig{0}$-formula with no free variables
 other than $x$, $y$ and $X$ expressing that
$X \subseteq \BF{\ell} \times \TC{\ell}$
is a set of pairs of terms such that $\seq{x, y} \in X$,
and, for any $\seq{t, s} \in X$,
$\qint{s} \leq \qint{t} \leq \ell$ and
$\exists l \rightarrow r \in \RS$,
$\exists \theta: V_{\RS} \rightarrow \TC{\ell}$ s.t.
$t = l \theta$ and one of the following cases holds.
  \begin{enumerate}
   \item $s = r \theta \in \TC{\ell}$.
	 \label{t:mfg:1}
   \item $\exists \left\langle \seq{t_j, s_j} \in X \mid j < \size{r} \right\rangle$ s.t.
         $s = 
         \big( (r \theta) [s_0/ t_0] \cdots \big)
	 [s_{\size{r} -1}/t_{\size{r} -1}]$,
	 where $s'[u/v]$ is identical if no $v$ occurs in $s'$.
	 \label{t:mfg:2}
  \end{enumerate}
Note that, since $V_{\RS}$ is a finite set of variables,
$\exists \theta: V_{\RS} \rightarrow \TC{\ell}$ 
can be regarded as a (first order) bounded quantifier.
By Proposition \ref{p:PQI}.\ref{p:PQI:1},
we can find a polynomial term
$p(x)$ such that $\qint{t} \leq p(|\code{t}|)$ holds for any
 $t \in \BF{}$.
The rest of the proof is devoted to deduce
$(\forall t \in \BF{}) (\exists s \in \TC{\ell})
 \exists G \ \psi_{p(|\code{t}|)} (t, s, G)$
 for such a bounding polynomial $p$.
 Fix an input basic term $t_0 \in \BF{}$ and let
 $\varphi_{\ell} (t)$ denote the $\BSig{1}$-formula
$(\exists s \in \TC{\ell}) \exists G \ \psi_{\ell} (t, s, G)$,
where $\ell = p(|\code{t_0}|)$.
Since $t_0 \in \BF{\ell}$, it suffices to deduce
$(\forall t \in \BF{\ell}) \varphi_{\ell} (t)$.
 By Lemma \ref{l:TI:B}, this follows from
$(\forall t \in \BF{\ell})
\big( (\forall s \in \BF{\ell}) (s <_{\ell} t \rightarrow \varphi_{\ell} (s))
       \rightarrow \varphi_{\ell}(t)
\big)
$,
which is the premise of an instance of
$\eqref{e:TI:B}$.
Thus let $t \in \BF{\ell}$ and assume the condition
\begin{equation} 
\label{e:t:prog}
(\forall s \in \BF{\ell}) (s <_{\ell} t \rightarrow \varphi_{\ell} (s)).
\end{equation}
Since $\RS$ is quasi-reducible,
there exist a rule $l \rightarrow r \in \RS$ and a substitution
$\theta: V_{\RS} \rightarrow \TC{\ell}$ such that
 $t = l \theta$.
The remaining argument splits into two cases depending on the shape of
$r \theta$.

{\sc Case \ref{t:mfg:1}.} 
 $r \theta \in \TC{\ell}$:
In this case
$\psi_{\ell} (t, r \theta, G)$ holds for the singleton
$G := \{ \seq{t, r \theta}\}$.

{\sc Case \ref{t:mfg:2}.} 
 $r \theta \not\in \TC{\ell}$:
 In this case there exists a basic subterm $v_0$ of $r \theta$.
 Fix a term $u_0 \in \TC{\ell}$ such that
 $\qint{u_0} \leq \qint{v_0}$.
 We show the following claim by finitary induction on
 $m < \size{r}$.
  \begin{claim}
  There exists a sequence
  $\left\langle \seq{t_j, s_j, G_j}
               \mid j \leq m
 \right\rangle$
 of triplets such that, for each $j \leq m$,
  {\normalfont (i)} $t_j <_{\ell} t$,
  {\normalfont (ii)} $\psi_{\ell} (t_j, s_j, G_j)$ holds, and
  {\normalfont (iii)}
  $\big( (r \theta) [s_0 / t_0] \cdots
  \big) [s_j / t_j]$
  is not identical to
  $\big( (r \theta) [s_0 / t_0] \cdots
  \big) [s_{j-1} / t_{j-1}]$
  as long as 
  $\big( (r \theta) [s_0 / t_0] \cdots
  \big) [s_{j-1} / t_{j-1}]$
  has a basic subterm.
  \end{claim}
 In the base case $m=0$, let $t_0$ be an arbitrary basic subterm of $r \theta$.
 Then, since $\qint{r \theta} \leq \qint{l \theta}$, $t_0 <_{\ell} t$
 follows from the definition of LPOs. 
 Hence, by the assumption \eqref{e:t:prog}, there exist a term
 $s_0 \in \TC{\ell}$ and a set $G_0$ such that
 $\psi_{\ell} (t_0, s_0, G_0)$ holds.
 Clearly, $(r \theta)[s_0/t_0]$ is not identical to $r \theta$.
 For induction step, suppose that there exists a sequence
  $\left\langle \seq{t_j, s_j, G_j}
               \mid j \leq m
 \right\rangle$
 fulfilling the conditions (i)--(iii) in the claim.
 In case that
  $\big( (r \theta) [s_0 / t_0] \cdots
  \big) [s_m / t_m]$
 has no basic subterm, let
 $(t_{m+1}, s_{m+1}) = (v_0, u_0)$.
 Otherwise, let $t_{m+1}$ be an arbitrary basic subterm.
 Then $t_{m+1} <_{\ell} t$ holds by Lemma
 \ref{l:sub}.
 Hence, as in the base case, the assumption \eqref{e:t:prog} yields a term
 $s_{m+1} \in \TC{\ell}$ and a set $G_{m+1}$ such that
 $\psi_{\ell} (t_{m+1}, s_{m+1}, G_{m+1})$ holds.
 By the choice of $t_{m+1}$,
  $\big( (r \theta) [s_0 / t_0] \cdots
 \big) [s_{m+1} / t_{m+1}]$
 is not identical to  
  $\big( (r \theta) [s_0 / t_0] \cdots
 \big) [s_m / t_m]$.

Now let
 $s:=
  \big( (r \theta) [s_0 / t_0] \cdots
 \big) [s_{\size{r} -1} / t_{\size{r} -1}]$
for a sequence
  $\left\langle \seq{t_j, s_j, G_j}
               \mid j < \size{r}
 \right\rangle$
witnessing the claim in case $m= \size{r} -1$. 
Then $s \in \TC{\ell}$
since
 $|\{ f \in \DS \mid f \text{ appears in }
      \big( (r \theta) [s_0 / t_0] \cdots
      \big) [s_j / t_j]
 \}| \leq \size{r} - (j+1)$
 holds for each $j < \size{r}$
  by the condition (iii) in the claim.
Defining a set $G$ by
$G = \{ \seq{t, s} \} \cup 
 \left( \bigcup_{j < \size{r}} G_j \right)$
now allows us to conclude 
$\psi_{\ell} (t, s, G)$.
\end{proof}

%% file: application.tex
\section{Application}
\label{s:application}

In the last section, to convince readers that the formalization of
termination proofs described in Theorem
\ref{t:MFG} for
$\LPOPoly$-programs is optimal, we show that the formalization
yields an alternative proof of Theorem \ref{t:pqi}, i.e., that
$\LPOPoly$-programs can only compute polynomial-space computable functions.

The next lemma ensures that the set $G$ constructed in Theorem
\ref{t:MFG} is indeed a minimal function graph.
 \begin{lemma}
  \label{l:MFG}
  Suppose that $\RS$ is a quasi-reducible $\LPOPoly$-program.
  Let $\psi_{\ell} (x, y, X)$ denote the $\BSig{0}$-formula defined in
  the proof of Theorem \ref{t:MFG}.
  Then, for any $t \in \BF{}$ and for any $t \in \TC{}$,
  $t \rewnf s$ if and only if
  $\exists G \ \psi_{p(|\code{t}|)} (t, s, G)$
  holds under the standard semantics.
 \end{lemma}

 \begin{proof}
  Let $\RS$ reduce under an LPO $<_{\lpo}$.
  For the  ``if'' direction, it can be shown that
  $(\forall t \in \BF{}) (\forall s \in \TC{})
  \big( \exists G \ \psi_{p(|\code{t}|)} (t, s, G) \Rightarrow  t \rewnf s
  \big)$
  holds by (external) transfinite induction along $<_{\lpo}$.
  For the ``only if'' direction, it can be shown that
  $(\forall t \in \BF{}) (\forall s \in \TC{})
  \big( t \rewm{m} s \Rightarrow \exists G \ \psi_{p(|\code{t}|)} (t, s, G) 
  \big)$
  holds by induction on $m$, where $\rewm{m}$ denotes the $m$-fold
  iteration of $\rew$.
 \end{proof}

Now Theorem \ref{t:MFG} and Lemma \ref{l:MFG} yield an alternative proof of
 (a variant of) Theorem \ref{t:pqi}.
 \begin{corollary}
   \label{c:pqi}
Every function computed by a quasi-reducible
$\LPOPoly$-program is computable in  polynomial space.
\end{corollary}

 \begin{proof}
  By Theorem \ref{t:MFG}, $\mU^1_2$ proves the formula
  \begin{equation*}
  \QR{} \wedge
  \LPO \wedge
  \PQI{} 
  \rightarrow 
  (\forall t \in \BF{}) (\exists s \in \TC{p(|\code{t}|)})
  \exists G \ \psi_{p(|\code{t}|)} (t, s, G),
  \end{equation*}
  where $\QR{}$, $\LPO{}$ and $\PQI{}$ respectively express that
  any $\BF{}$-term is reducible,
  $\RS$ reduces under $<_{\lpo}$, and
  $(\forall (l \rightarrow r) \in \RS)
  (\forall \theta: V_{\RS} \rightarrow \TC{})
  \qint{r \theta} \leq \qint{l \theta}$.
By Lemma \ref{l:lpoU}, $\LPO$ can be expressed with a
  $\BSig{0}$-formula, but neither $\QR{}$ nor $\PQI{}$ is literally
  expressible with a bounded formula.
Nonetheless, the proof can be easily modified to a proof of the statement 
\begin{equation*}
 (\forall t \in \BF{}) (\exists s \in \TC{\ell})
  \big(
  \QR{\ell} \wedge
  \LPO \wedge
  \PQI{\ell} 
  \rightarrow 
  \exists G \ \psi_{\ell} (t, s, G)
  \big),
\end{equation*}
  where $\ell = p(|\code{t}|)$, and $\QR{\ell}$ and $\PQI{\ell}$
  respectively express that
  any $\BF{\ell}$-term is reducible, and
  $(\forall (l \rightarrow r) \in \RS)
  (\forall \theta: V_{\RS} \rightarrow \TC{\ell})
  \qint{r \theta} \leq \qint{l \theta}$.
  Both $\QR{\ell}$ and $\PQI{\ell}$ can be regarded as
  $\BSig{0}$-formulas, and hence the formula
  $\varphi_{\ell} (t, s) :\equiv
  \QR{\ell} \wedge
  \LPO \wedge
  \PQI{\ell} 
  \rightarrow 
  \exists G \ \psi_{\ell} (t, s, G)
  $
  lies in $\BSig{1}$.

  Now suppose that a function
  $\fint{\mf}: \TC{}^k \rightarrow \TC{}$
  is computed by a quasi-reducible $\LPOPoly$-program $\RS$
  for some $k$-ary function symbol $\mf \in \DS$.
  Then Lemma \ref{l:Buss} yields a polynomial-space computable function
  $f: \mathbb N^k \rightarrow \mathbb N$ such that
  $\varphi_{p(|\code{\mf (t_1, \dots, t_k}|)}
   \big( \mf (t_1, \dots, t_k),
         \underline{f (\code{t_1}, \dots, \code{t_k})}
   \big)$
  holds for any
  $t_1, \dots, t_k \in \TC{}$ under the standard semantics.
  Hence, by assumption,
  $\psi_{p(|\code{\mf (t_1, \dots, t_k}|)}
   \big( \mf (t_1, \dots, t_k),
         \underline{f (\code{t_1}, \dots, \code{t_k})}, G
   \big)$
  holds for some set $G \subseteq \BF{} \times \TC{}$.
  By Lemma \ref{l:MFG}, this means the correspondence
  $\fint{\mf} (t_1, \dots, t_k) = s \Leftrightarrow
  f(\code{t_1}, \dots, \code{t_k}) = \code{s}$.
  Therefore,
$\code{\fint{\mf} (t_1, \dots, t_k)}$ can be computed with space bounded by a
 polynomial in $|\code{t_1}|, \dots, |\code{t_k}|$
and thus bounded by a polynomial in $\size{t_1}, \dots, \size{t_k}$.
 \end{proof}

%% file: paper.bbl
\begin{thebibliography}{10}
\providecommand{\bibitemdeclare}[2]{}
\providecommand{\surnamestart}{}
\providecommand{\surnameend}{}
\providecommand{\urlprefix}{Available at }
\providecommand{\url}[1]{\texttt{#1}}
\providecommand{\href}[2]{\texttt{#2}}
\providecommand{\urlalt}[2]{\href{#1}{#2}}
\providecommand{\doi}[1]{doi:\urlalt{http://dx.doi.org/#1}{#1}}
\providecommand{\bibinfo}[2]{#2}

\bibitemdeclare{article}{BeckB14}
\bibitem{BeckB14}
\bibinfo{author}{A.~\surnamestart Beckmann\surnameend} \& \bibinfo{author}{S.R.
  \surnamestart Buss\surnameend} (\bibinfo{year}{2014}):
  \emph{\bibinfo{title}{Improved Witnessing and Local Improvement Principles
  for Second-order Bounded Arithmetic}}.
\newblock {\sl \bibinfo{journal}{ACM Transactions on Computational Logic}}
  \bibinfo{volume}{15}(\bibinfo{number}{1}), p.~\bibinfo{pages}{2},
  \doi{10.1145/2559950}.

\bibitemdeclare{article}{BCMT01}
\bibitem{BCMT01}
\bibinfo{author}{G.~\surnamestart Bonfante\surnameend},
  \bibinfo{author}{A.~\surnamestart Cichon\surnameend}, \bibinfo{author}{J.-Y.
  \surnamestart Marion\surnameend} \& \bibinfo{author}{H.~\surnamestart
  Touzet\surnameend} (\bibinfo{year}{2001}): \emph{\bibinfo{title}{Algorithms
  with Polynomial Interpretation Termination Proof}}.
\newblock {\sl \bibinfo{journal}{Journal of Functional Programming}}
  \bibinfo{volume}{11}(\bibinfo{number}{1}), pp. \bibinfo{pages}{33--53},
  \doi{10.1017/S0956796800003877}.

\bibitemdeclare{inproceedings}{BMM01}
\bibitem{BMM01}
\bibinfo{author}{G.~\surnamestart Bonfante\surnameend}, \bibinfo{author}{J.-Y.
  \surnamestart Marion\surnameend} \& \bibinfo{author}{J.-Y. \surnamestart
  Moyen\surnameend} (\bibinfo{year}{2001}): \emph{\bibinfo{title}{{On
  Lexicographic Termination Ordering with Space Bound Certifications}}}.
\newblock In: {\sl \bibinfo{booktitle}{Perspectives of System Informatics}},
  {\sl \bibinfo{series}{Lecture Notes in Computer Science}}
  \bibinfo{volume}{2244}, pp. \bibinfo{pages}{482--493},
  \doi{10.1007/3-540-45575-2\_46}.

\bibitemdeclare{article}{BMM11}
\bibitem{BMM11}
\bibinfo{author}{G.~\surnamestart Bonfante\surnameend}, \bibinfo{author}{J.-Y.
  \surnamestart Marion\surnameend} \& \bibinfo{author}{J.-Y. \surnamestart
  Moyen\surnameend} (\bibinfo{year}{2011}):
  \emph{\bibinfo{title}{Quasi-interpretations A Way to Control Resources}}.
\newblock {\sl \bibinfo{journal}{Theoretical Computer Science}}
  \bibinfo{volume}{412}(\bibinfo{number}{25}), pp. \bibinfo{pages}{2776--2796},
  \doi{10.1016/j.tcs.2011.02.007}.

\bibitemdeclare{article}{buch95}
\bibitem{buch95}
\bibinfo{author}{W.~\surnamestart Buchholz\surnameend} (\bibinfo{year}{1995}):
  \emph{\bibinfo{title}{Proof-theoretic Analysis of Termination Proofs}}.
\newblock {\sl \bibinfo{journal}{Annals of Pure and Applied Logic}}
  \bibinfo{volume}{75}(\bibinfo{number}{1--2}), pp. \bibinfo{pages}{57--65},
  \doi{10.1016/0168-0072(94)00056-9}.

\bibitemdeclare{book}{buss86}
\bibitem{buss86}
\bibinfo{author}{S.R. \surnamestart Buss\surnameend} (\bibinfo{year}{1986}):
  \emph{\bibinfo{title}{Bounded Arithmetic}}.
\newblock \bibinfo{publisher}{Bibliopolis, Napoli}.

\bibitemdeclare{incollection}{buss98}
\bibitem{buss98}
\bibinfo{author}{S.R. \surnamestart Buss\surnameend} (\bibinfo{year}{1998}):
  \emph{\bibinfo{title}{{First-Order Proof Theory of Arithmetic}}}.
\newblock In \bibinfo{editor}{S.R. \surnamestart Buss\surnameend}, editor: {\sl
  \bibinfo{booktitle}{Handbook of Proof Theory}}, \bibinfo{publisher}{North
  Holland, Amsterdam}, pp. \bibinfo{pages}{79--147},
  \doi{10.1016/S0049-237X(98)80017-7}.

\bibitemdeclare{article}{Dershowitz82}
\bibitem{Dershowitz82}
\bibinfo{author}{N.~\surnamestart Dershowitz\surnameend}
  (\bibinfo{year}{1982}): \emph{\bibinfo{title}{{Orderings for Term-Rewriting
  Systems}}}.
\newblock {\sl \bibinfo{journal}{Theoretical Computer Science}}
  \bibinfo{volume}{17}, pp. \bibinfo{pages}{279--301},
  \doi{10.1016/0304-3975(82)90026-3}.

\bibitemdeclare{incollection}{Eguchi10}
\bibitem{Eguchi10}
\bibinfo{author}{N.~\surnamestart Eguchi\surnameend} (\bibinfo{year}{2010}):
  \emph{\bibinfo{title}{A Term-rewriting Characterization of PSPACE}}.
\newblock In \bibinfo{editor}{T.~\surnamestart Arai\surnameend},
  \bibinfo{editor}{C.T. \surnamestart Chong\surnameend},
  \bibinfo{editor}{R.~\surnamestart Downey\surnameend},
  \bibinfo{editor}{J.~\surnamestart Brendle\surnameend},
  \bibinfo{editor}{Q.~\surnamestart Feng\surnameend},
  \bibinfo{editor}{H.~\surnamestart Kikyo\surnameend} \&
  \bibinfo{editor}{H.~\surnamestart Ono\surnameend}, editors: {\sl
  \bibinfo{booktitle}{Proceedings of the 10th Asian Logic Conference 2008}},
  \bibinfo{publisher}{World Scientific}, pp. \bibinfo{pages}{93--112},
  \doi{10.1142/9789814293020\_0004}.

\bibitemdeclare{inproceedings}{Hof90}
\bibitem{Hof90}
\bibinfo{author}{D.~\surnamestart Hofbauer\surnameend} (\bibinfo{year}{1990}):
  \emph{\bibinfo{title}{Termination Proofs by Multiset Path Orderings Imply
  Primitive Recursive Derivation Lengths}}.
\newblock In: {\sl \bibinfo{booktitle}{Proceedings of the 2nd International
  Conference on Algebraic and Logic Programming}}, {\sl
  \bibinfo{series}{Lecture Notes in Computer Science}} \bibinfo{volume}{463},
  pp. \bibinfo{pages}{347--358}, \doi{10.1007/3-540-53162-9\_50}.

\bibitemdeclare{book}{Jones97}
\bibitem{Jones97}
\bibinfo{author}{N.D. \surnamestart Jones\surnameend} (\bibinfo{year}{1997}):
  \emph{\bibinfo{title}{Computability and Complexity - from a Programming
  Perspective}}.
\newblock \bibinfo{series}{Foundations of Computing Series},
  \bibinfo{publisher}{MIT Press}, \doi{10.1007/978-94-010-0413-8\_4}.

\bibitemdeclare{inproceedings}{JonesM86}
\bibitem{JonesM86}
\bibinfo{author}{N.D. \surnamestart Jones\surnameend} \&
  \bibinfo{author}{A.~\surnamestart Mycroft\surnameend} (\bibinfo{year}{1986}):
  \emph{\bibinfo{title}{Data Flow Analysis of Applicative Programs Using
  Minimal Function Graphs}}.
\newblock In: {\sl \bibinfo{booktitle}{Proceedings of the 13th ACM Symposium on
  Principles of Programming Languages}}, pp. \bibinfo{pages}{296--306},
  \doi{10.1145/512644.512672}.

\bibitemdeclare{unpublished}{Kamin_Levy}
\bibitem{Kamin_Levy}
\bibinfo{author}{S.~\surnamestart Kamin\surnameend} \& \bibinfo{author}{J.-J.
  \surnamestart L\'{e}vy\surnameend} (\bibinfo{year}{1980}):
  \emph{\bibinfo{title}{Two Generalizations of the Recursive Path Ordering}}.
\newblock \bibinfo{note}{Unpublished manuscript, University of Illinois}.

\bibitemdeclare{article}{LM95}
\bibitem{LM95}
\bibinfo{author}{D.~\surnamestart Leivant\surnameend} \& \bibinfo{author}{J.-Y.
  \surnamestart Marion\surnameend} (\bibinfo{year}{1995}):
  \emph{\bibinfo{title}{{Ramified Recurrence and Computational Complexity II:
  Substitution and Poly-space}}}.
\newblock {\sl \bibinfo{journal}{Lecture Notes in Computer Science}}
  \bibinfo{volume}{933}, pp. \bibinfo{pages}{486--500},
  \doi{10.1007/BFb0022277}.

\bibitemdeclare{article}{Marion03}
\bibitem{Marion03}
\bibinfo{author}{J.-Y. \surnamestart Marion\surnameend} (\bibinfo{year}{2003}):
  \emph{\bibinfo{title}{Analysing the Implicit Complexity of Programs}}.
\newblock {\sl \bibinfo{journal}{Information and Computation}}
  \bibinfo{volume}{183}(\bibinfo{number}{1}), pp. \bibinfo{pages}{2--18},
  \doi{10.1016/S0890-5401(03)00011-7}.

\bibitemdeclare{incollection}{Oitavem01}
\bibitem{Oitavem01}
\bibinfo{author}{I.~\surnamestart Oitavem\surnameend} (\bibinfo{year}{2001}):
  \emph{\bibinfo{title}{Implicit Characterizations of Pspace}}.
\newblock In: {\sl \bibinfo{booktitle}{Proof Theory in Computer Science}}, {\sl
  \bibinfo{series}{Lecture Notes in Computer Science}} \bibinfo{volume}{2183},
  \bibinfo{publisher}{Springer}, pp. \bibinfo{pages}{170--190},
  \doi{10.1007/3-540-45504-3\_11}.

\bibitemdeclare{article}{Oitavem02}
\bibitem{Oitavem02}
\bibinfo{author}{I.~\surnamestart Oitavem\surnameend} (\bibinfo{year}{2002}):
  \emph{\bibinfo{title}{A Term Rewriting Characterization of the Functions
  Computable in Polynomial Space}}.
\newblock {\sl \bibinfo{journal}{Archive for Mathematical Logic}}
  \bibinfo{volume}{41}(\bibinfo{number}{1}), pp. \bibinfo{pages}{35--47},
  \doi{10.1007/s001530200002}.

\bibitemdeclare{article}{Savitch70}
\bibitem{Savitch70}
\bibinfo{author}{W.J. \surnamestart Savitch\surnameend} (\bibinfo{year}{1970}):
  \emph{\bibinfo{title}{Relationships Between Nondeterministic and
  Deterministic Tape Complexities}}.
\newblock {\sl \bibinfo{journal}{Journal of Computer and System Sciences}}
  \bibinfo{volume}{4}(\bibinfo{number}{2}), pp. \bibinfo{pages}{177--192},
  \doi{10.1016/S0022-0000(70)80006-X}.

\bibitemdeclare{book}{SlonnegerK95}
\bibitem{SlonnegerK95}
\bibinfo{author}{K.~\surnamestart Slonneger\surnameend} \&
  \bibinfo{author}{B.L. \surnamestart Kurtz\surnameend} (\bibinfo{year}{1995}):
  \emph{\bibinfo{title}{Formal Syntax and Semantics of Programming Languages -
  A Laboratory Based Approach}}.
\newblock \bibinfo{publisher}{Addison-Wesley}.

\bibitemdeclare{book}{Terese}
\bibitem{Terese}
\bibinfo{author}{\surnamestart Terese\surnameend} (\bibinfo{year}{2003}):
  \emph{\bibinfo{title}{Term Rewriting Systems}}.
\newblock {\sl \bibinfo{series}{Cambridge Tracts in Theoretical Computer
  Science}}~\bibinfo{volume}{55}, \bibinfo{publisher}{Cambridge University
  Press}.

\bibitemdeclare{article}{Thompson72}
\bibitem{Thompson72}
\bibinfo{author}{D.B. \surnamestart Thompson\surnameend}
  (\bibinfo{year}{1972}): \emph{\bibinfo{title}{Subrecursiveness:
  Machine-Independent Notions of Computability in Restricted Time and
  Storage}}.
\newblock {\sl \bibinfo{journal}{Mathematical Systems Theory}}
  \bibinfo{volume}{6}(\bibinfo{number}{1}), pp. \bibinfo{pages}{3--15},
  \doi{10.1007/BF01706069}.

\bibitemdeclare{article}{Weier95}
\bibitem{Weier95}
\bibinfo{author}{A.~\surnamestart Weiermann\surnameend} (\bibinfo{year}{1995}):
  \emph{\bibinfo{title}{Termination Proofs for Term Rewriting Systems by
  Lexicographic Path Orderings Imply Multiply Recursive Derivation Lengths}}.
\newblock {\sl \bibinfo{journal}{Theoretical Computer Science}}
  \bibinfo{volume}{139}(\bibinfo{number}{1{\&}2}), pp.
  \bibinfo{pages}{355--362}, \doi{10.1016/0304-3975(94)00135-6}.

\bibitemdeclare{book}{Winskel93}
\bibitem{Winskel93}
\bibinfo{author}{G.~\surnamestart Winskel\surnameend} (\bibinfo{year}{1993}):
  \emph{\bibinfo{title}{The Formal Semantics of Programming Languages - An
  Introduction}}.
\newblock \bibinfo{series}{Foundations of Computing Series},
  \bibinfo{publisher}{MIT Press}.

\end{thebibliography}
